\documentclass{article}

\usepackage{natbib}

\title{Graphical Potential Games}

\usepackage{amsfonts}
\usepackage{bbm}

\usepackage{amsmath,amssymb,amsthm}
\usepackage{url}
\usepackage{hyperref}

\providecommand{\R}{\mathbb{R}} 
\providecommand{\KL}{\text{KL}} 
\providecommand{\Pa}{\text{Pa}} 
\providecommand{\Nd}{\text{Nd}} 
\DeclareMathOperator*{\argmax}{arg\,max}
\DeclareMathAlphabet{\mathpzc}{OT1}{pzc}{m}{it}

\newtheorem{definition}{Definition}
\newtheorem{theorem}{Theorem}
\newtheorem{lemma}{Lemma}
\newtheorem{corollary}{Corollary}
\newtheorem{proposition}{Proposition}
\newtheorem{property}{Property}
\newtheorem{remark}{Remark}{\bfseries}{\rmfamily}

\newcommand{\indicator}[1]{\mathbbm{1}{\left[ {#1} \right] }}

\author{Luis E. Ortiz\\Department of Computer Science\\Stony Brook
  University\\Stony Brook, NY  11794-4400}

\begin{document}

\maketitle              

\begin{abstract}
  Potential games, originally introduced in the early 1990's by Lloyd Shapley, the 2012 Nobel Laureate in Economics, and his colleague Dov Monderer, are a very important class of models in game theory. They have special properties such as the existence of Nash equilibria in pure strategies. This note introduces graphical versions of potential games. Special cases of graphical potential games have already found applicability in many areas of science and engineering beyond economics, 
including artificial intelligence, computer vision, and machine learning.
They have been effectively applied to the study and solution of important real-world problems such as routing and congestion in networks, distributed resource allocation (e.g., public goods), and relaxation-labeling for image segmentation.
Implicit use of graphical potential games goes back at least 40 years. Several classes of games considered standard in the literature, including coordination games, local interaction games, lattice games, congestion games, and party-affiliation games, are instances of graphical potential games. This note provides several characterizations of graphical potential games by leveraging well-known results from the literature on probabilistic graphical models.  A major contribution of the work presented here that particularly distinguishes it from previous work is establishing that the convergence of certain type of game-playing rules implies that the agents/players must be embedded in some graphical potential game.
\end{abstract}
\section{Introduction}

Potential games~\citep{monderer96} have become an inherently important  class of models in game theory. Potential games have special properties such as the existence of Nash equilibria in pure strategies. By now, potential games are so fundamental and core to game theory that their
study is of broad interest.

This note introduces \emph{graphical potential games}, a graphical
version of classical potential games.
Implicit use of graphical potential games goes back at least 40 years. Special cases of graphical potential games have already found applicability in many areas of science and engineering beyond economics. These areas include artificial intelligence and
computer vision~\citep{Miller_and_Zucker_1991,Yu_and_Berthod_1995,Berthod_et_al_96}, machine learning~\citep{rezek08}, neural networks~\citep{hopfield82,Miller_and_Zucker_1991}, 
theoretical computer science, 
computational social science and sociology~\citep{Montanari09,Montanari12112010}, and dynamical systems ~\citep{Miller_and_Zucker_1991,zucker01}. They have been effectively applied to the study and solution of important real-world problems such as routing and congestion in networks~\citep{rosenthal73}, distributed resource allocation such as public goods~\citep{heikkinen06}, relaxation-labeling for image segmentation~\citep{rosenfeld76,Miller_and_Zucker_1991,zucker01,Yu_and_Berthod_1995,Berthod_et_al_96}, clustering and probabilistic inference in graphical models via learning in games~\citep{rezek08}, equilibrium-selection in strategic settings~\citep{young93}, the emergence of coordination from individual agent's strategic behavior~\citep{Blume1995111}, the rate at which innovations and norms/conventions spread through a social network and the role of the network structure~\citep{Montanari09,Montanari12112010}, and inference in Hopfield networks~\citep{hopfield82}. (Please refer to the given reference for details.)
Several classes of games considered standard in the literature, including coordination games, local interaction games~\citep{Montanari09,Montanari12112010}, lattice games~\citep{Blume1993387}, congestion games~\citep{rosenthal73}, and party-affiliation games~\citep{fabrikant04}, are instances of graphical potential games.

Potential games, like normal-form games,
do not have an inherently compact, graphical representation, in contrast to graphical games~\citep{kearns01}. Not all potential games are non-trivial graphical games.
In a graphical game, the payoff function of each player is a function
of its neighbors in the game graph. Any definition of a graphical game
must respect that. But the potential function of a potential game is
\emph{global}: it
involves \emph{all} players! That means that we can only use certain types of potential
functions to define a graphical potential game. 

This note provides several, strong characterizations of graphical
potential games by leveraging well-known results from the literature
on probabilistic graphical models.  A major contribution of the work
presented here that particularly distinguishes it from previous work
is establishing that the convergence of certain type of game-playing
rules implies that the agents/players must be embedded in some
graphical potential game. At this point, it is best to 
delay the discussion of the most closely related work in economics and, more recently, computer science until the end of the paper.

\section{Preliminaries}

This section introduces basic notation and concepts in graphical models and game theory.

\paragraph{Basic Notation.} Denote by $x \equiv (x_1,x_2,\ldots,x_n)$ an $n$-dimensional vector and by $x_{-i} \equiv (x_1,\ldots,x_{i-1},x_{i+1},\ldots,x_n)$ the same vector without component $i$. Similarly, for every set $S \subset [n] \equiv \{1,\ldots,n\}$, denote by $x_S \equiv (x_i : i \in S)$ the (sub-)vector formed from $x$ using only components in $S$, such that, letting $S^c \equiv [n] - S$ denote the complement of $S$, we can denote $x \equiv (x_S,x_{S^c}) \equiv (x_i,x_{-i})$ for every $i$. If $A_1,\ldots,A_n$ are sets, denote by $A \equiv \times_{i \in [n]} A_i$, $A_{-i} \equiv \times_{j \in [n] - \{i\}} A_j$ and $A_S \equiv \times_{j \in S} A_j$.

\paragraph{Graph Terminology and Notation.} Let $G = (V,E)$ be an undirected graph,
with finite set of $n$ {\em vertices\/} or {\em nodes\/} $V = [n]$
and a set of (undirected) edges $E$.
For each node $i$, let $\mathcal{N}(i) \equiv \{ j \mid (i,j) \in E \}$ be the set of neighbors of $i$ in $G$, {\em not including\/} $i$, and $N(i)  \equiv \mathcal{N}(i) \cup  \{i\}$ the set {\em including\/} $i$.
A {\em clique\/} $C$
of $G$ is a set of nodes with the property that they are all mutually connected: for all $i, j \in C$, $(i,j) \in E$; in addition, $C$ is {\em maximal\/} if there is no other node $k$ outside $C$ that is also connected to each node in $C$, i.e., for all $k \in V - C$, $(k,i) \notin E$ for some $i \in C$.

Another useful concept in the context of this note is that of hypergraphs, which are generalizations of regular graphs. A {\em hypergraph\/} $\mathcal{G} = (V,\mathcal{E})$ is defined by a set of nodes $V$ and a set of {\em hyperedges\/} $\mathcal{E} \subset 2^V$. We can think of the hyperedges as cliques in a regular graph. 

\subsection{Graphical Models}

Graphical models~\citep{kollerandfriedman09} are an elegant marriage of statistics and graph theory that has had tremendous impact in the theory and practice of modern statistics. It has permitted effective modeling of large, structured high-dimensional complex systems found in the real world. The language of graphical models allows us to capture the probabilistic structure of complex interactions between individual entities in the system. The core component of the model is a graph in which each node $i$ corresponds to a random variable $X_i$ and missing edges express \emph{conditional independence assumptions} about those random variables in the probabilistic system. 

\subsubsection{Markov Random Fields, Gibbs Distributions and the Hammersley-Clifford Theorem} 
A joint probability distribution $P$ is called a {\em Markov random field (MRF)\/} with respect to an undirected graph $G$ if for all $x$, for every node $i$,
\(
P(X_i = x_i \mid X_{-i} = x_{-i}) = P(X_i = x_i \mid X_{\mathcal{N}(i)} = x_{\mathcal{N}(i)}).
\)
In that case, the neighbors/variables $X_{\mathcal{N}(i)}$ form the {\em Markov blanket\/} of node/variable $X_i$. 

A joint distribution $P$ is called a {\em Gibbs distribution\/} with respect to a an undirected graph $G$ if 
it can be expressed as 
\(
\textstyle
P(X = x) = \prod_{C \in \mathcal{C}} \Phi_C(x_C)
\)
for some functions $\Phi_C$ indexed by a clique $C \in \mathcal{C}$, the set of all (maximal) cliques in $G$, and mapping every possible value $x_C$ that the random variables $X_C$ associated with the nodes in $C$ can take to a non-negative number. 

Let us say that joint probability distribution $P$ is {\em positive\/} if it has full support (i.e., $P(x) > 0$ for all $x$).
\begin{theorem}
\label{thm:hc}
{\bf [Hammersley-Clifford]}~\citep{hammersley71} Let $P$ be a positive joint probability distribution. Then, $P$ is an MRF with respect to $G$ if and only if $P$ is a Gibbs distribution with respect to $G$.
\end{theorem}

\noindent In the context of the theorem, the functions $\Phi_C$ are positive, which allows us to define MRFs in terms of {\em local potential functions\/} $\{ \phi_C \}$ over each clique $C$ in the graph. 
Define the function $\Psi(x) \equiv \sum_{C \in \mathcal{C}} \phi_C(x_C)$. Let us refer to any function of this form as a {\em Gibbs potential\/} with respect to $G$. Thanks to the Hammersley-Clifford Theorem, a more familiar expression of an MRF is
\(
\textstyle
P(X = x) \propto \exp(\sum_{C \in \mathcal{C}} \phi_C(x_C)) = \exp(\Psi(x)) .
\)

\subsection{Game Theory and Graphical Games}

Let $V = [n]$
denote a finite set of $n$ players in a game. For each player $i \in V$, let $A_i$ denote the (finite) set of {\em actions} or {\em pure strategies\/} that $i$ can play. Let $A$
denote the set of {\em joint actions\/}, $x \equiv (x_i, \ldots, x_n) \in A$ denote a joint action, and $x_i$ the individual action of player $i$ in $x$. Denote by $x_{-i}$
the joint action of all the players {\em except\/} $i$, such that $x \equiv (x_i, x_{-i})$. Let $M_i : A \to \R$ denote the {\em payoff function\/} of player $i$. 

There are a variety of compact representations for large game inspired by probabilistic graphical models in AI and machine learning~\citep{la-mura00,kearns01,koller03,leyton-brown03,jiang08}. This paper's context is 
{\em graphical games}~\citep{kearns01}, a simple but powerful model inspired by probabilistic graphical models such as MRFs.
 
A \emph{graphical game}~\citep{kearns01} is defined by an undirected graph $(V,E)$ and a set of local payoff hypermatrices $\{M'_i : A_{N(i)} \to \R \mid i  \in V \}$. Each node $i \in V$ in the graph corresponds to a player in the game and its payoff function $M_i(x) \equiv M'_i(x_{N(i)})$ is defined by its local payoff hypermatrix $M'_i$, a function of the actions of the players in the neighborhood $N(i)$ of $i$ only, which includes $i$. A {\em hypergraphical game\/}~\citep{papadimitriou05_ce} is defined by a hypergraph $(V,\mathcal{E})$ and sets of ``local'' payoffs hypermatrices of the form $\{M'_{i,C} : A_C \to \R  \mid C \in \mathcal{E}, i \in C \}$. Let $\mathcal{C}_i \equiv \{ C \in \mathcal{E} \mid  i \in C\}$ be the set of cliques in which $i$ participates. The payoff function of player $i$ in a hypergraphical game is defined as $M_i(x) \equiv \sum_{C \in \mathcal{C}_i} M'_{i,C}(x_C)$. Note that the payoff function of player $i$ only depends directly on the actions of players in its neighborhood $N(i) = \cup_{C \in \mathcal{C}_i} C = \{ j \in V \mid i,j \in C \text{ for some } C \in \mathcal{E}\}$, which includes $i$.
A {\em polymatrix game\/}~\citep{janovskaja68} is a hypergraphical game in which $\mathcal{C} = \{ \{i,j\} \mid i,j \in V, j \neq i \}$, which is the set of cliques of pairs of nodes involving the player and every other player. If, instead, the hyperedge set of each player is some (possibly different) {\em subset\/} of $\{ \{i,j\} \mid j \in V, j \neq i \}$, then let us call the game a {\em graphical polymatrix game}. Finally, let us say that a hypergraphical game is {\em hyperedge-symmetric\/} if, in addition, for every hyperedge $C$ containing players $i,j$ in the game, we have that $M'_{i,C} = M'_{j,C} \equiv M'_C$; if, in particular, the hypergraphical game is a graphical polymatrix game, then let us say that the game is {\em pairwise-symmetric}.

\section{Graphical Potential Games}

This section introduces graphical potential games (and other related subclasses) and provides
some structural properties and characterizations of such games. Previous work in probabilistic graphical models facilitate the derivations of the results.

\begin{definition}
Consider a graph $G$ with (non-inclusive) neighbor sets $\mathcal{N}(i)$ for each payer $i$. 
For any graphical game with graph $G$, and for each player $i$, consider some function $f_i : \R \times A_{\mathcal{N}(i)} \to \R$.
Let us say that the function $f_i$ is a {\em 
(conditional) preference-order-preserving transform\/} for player $i$ if $f_i(v,x_{\mathcal{N}(i)})$ is a (strictly) monotonically increasing function of $v$ for every $x_{\mathcal{N}(i)} \in A_{\mathcal{N}(i)}$.~\footnote{\citet{monderer96} considered such transforms in the context of (variants of) fictitious play.} Let us call the transform {\em (unconditionally) linear\/} with respect to some positive weight $w_i$ if it takes the form $f_i(v,x_{\mathcal{N}(i)}) = w_i v$. Let us denote by $f^{-1}_i$ the corresponding {\em (conditional) inverse function\/} of $f_i$; that is, $f_i^{-1} : \R \times A_{\mathcal{N}(i)} \to \R$ such that for all $v \in \R$ and $x_{\mathcal{N}(i)} \in A_{\mathcal{N}(i)} $, we have $f_i^{-1}(f_i(v,x_{\mathcal{N}(i)}),x_{\mathcal{N}(i)})  = v$.
Denote by $f \equiv \{f_i\}$ and $f^{-1} \equiv \{f^{-1}_i\}$ the set
of preference-order-preserving transforms $f_i$ and their inverses,
respectively, one for each player $i$. Let us say that a function
$\Psi$ is a {\em $f$-transformed potential\/} for a graphical game
with neighbor sets $\{\mathcal{N}(i)\}$ and local payoff matrices
$\{M'_i\}$ if for all $i$, $x_{\mathcal{N}(i)}$, and $x_i,x'_i$, \( M'_i(x_i,x_{\mathcal{N}(i)}) - M'_i(x'_i,x_{\mathcal{N}(i)}) = f_i(\Psi(x_i,x_{-i}) - \Psi(x'_i,x_{-i}),x_{\mathcal{N}(i)}) \; .\) Let us call a graphical game with a $f$-transformed potential a {\em graphical $f$-transformed potential game}.
\end{definition}
We can generalize the terminology of \citet{monderer96} to graphical games. Let $w$ be a positive weight vector. If the transformed $f_i$ of each player $i$ is
(unconditionally) linear with weight $w_i$, 
then $\Psi$ is called a {\em (weighted) $w$-potential\/} for the game. 
If $w_i=1$ for all $i$, then $\Psi$ is called an {\em (exact) potential\/}.
Finally, let us say that $\Psi$ is an {\em ordinal potential\/} for a
graphical game with neighbor sets $\{\mathcal{N}(i)\}$ and local
payoff matrices $\{M'_i\}$ if it satisfy the following condition: \( M'_i(x_i,x_{\mathcal{N}(i)}) - M'_i(x'_i,x_{\mathcal{N}(i)}) > 0 \text{ if and only if } \Psi(x_i,x_{-i}) - \Psi(x'_i,x_{-i}) > 0\) for all $i, x_i,x'_i$ and $x_{\mathcal{N}(i)}$. Let us refer to a graphical game with a (weighted) $w$-potential, (exact) potential or ordinal potential as a {\em graphical (weighted) $w$-potential, (exact) potential or ordinal potential game}, respectively.


\section{Characterizing Graphical Potential Games}

The following theorem characterizes graphical $f$-transformed potential games: the potential function is the sum of local potential functions over each (maximal) clique in the game graph.
\begin{theorem}
\label{the:gibbspot}
Every potential of a graphical transformed potential game with graph $G$ is a Gibbs potential with respect to $G$.
\end{theorem}
\begin{proof}
The proof is based on an application of the Hammersley-Clifford Theorem (Theorem~\ref{thm:hc}). Let $\Psi$ be the $f$-transformed potential of the graphical game. Define $P$ as a joint probability distribution such that $P(X=x) \propto \exp(\Psi(x))$ for all $x$. (Thus, $P$ is positive.) The following derivation shows that $P$ is an MRF with respect to the game graph $G$: for all $i, x_i,x_{-i}$,
\begin{align*}
& P(X_i = x_i \mid X_{-i} = x_{-i})\\
 & =
\frac{\exp(\Psi(x_i,x_{-i}))}{\sum_{x'_i} \exp(\Psi(x'_i,x_{-i}))}\\ 
& = \frac{1}{\sum_{x'_i} \exp(\Psi(x'_i,x_{-i}) - \Psi(x_i,x_{-i}))}\\
& = \frac{1}{\sum_{x'_i} \exp(f^{-1}_i(M'_i(x'_i,x_{\mathcal{N}(i)}) - M'_i(x_i,x_{\mathcal{N}(i)}), x_{\mathcal{N}(i)}))} \; .
\end{align*}
Hence, by the Hammersley-Clifford Theorem (Theorem~\ref{thm:hc}), $P$ is also a Gibbs distribution with respect to the graph $G$ of the game. In particular, let $\mathcal{C}$ be the set of (maximal) cliques of $G$. There exist a local potential function $\phi'_C$ for each (maximal) clique $C \in \mathcal{C}$ defining a global potential function $\Psi'$ as $\Psi'(x) = \sum_{C \in \mathcal{C}} \phi'_C(x_C)$ such that $P(X=x) \propto \exp(\Psi'(x))$. Let $Z = \sum_x \exp(\Psi(x))$ and $Z' = \sum_x \exp(\Psi'(x))$ be the normalizing constant when expressing $P$ in terms of $\Psi$ and $\Psi'$, respectively, denote by $c \equiv \ln(Z/Z')$, a constant. Then we have that for all $x$, $\Psi(x) = \Psi'(x) + c$. Defining local potential for $\Psi$ as, for example, $\phi_C(x_C) \equiv \phi'_C(x) + c/|\mathcal{C}|$ completes the proof. 
\end{proof}

In view of the strong equivalence established by the last theorem, let us refer to a graphical transformed potential game as a transformed {\em Gibbs-potential game\/} (with the same graph); and similarly for weighted and exact potentials. 
Whether every graphical ordinal potential game has an (equivalent) ordinal Gibbs potential is left open. So, let us define an {\em ordinal Gibbs-potential game\/} as a graphical game that has a Gibbs potential with respect to the graph of the game.

The following definitions are useful to present the main corollary of the last theorem.
\begin{definition}
Given an $n$-dimensional positive weight vector $w$, let us say that a game with payoff hypermatrices $\{M^1_i\}$ is {\em $w$-scaled payoff-difference equivalent\/} to another game with the same players and payoff hypermatrices $\{M^2_i\}$ if for all $i, x_{-i}, x_i, x'_i$, we have that $M^1_i(x_i,x_{-i}) - M^1_i(x'_i,x_{-i}) = w_i (M^2_i(x_i,x_{-i}) - M^2_i(x'_i,x_{-i}))$. If $w$ is the vector of all $1$'s, then let us simply say that the game is {\em payoff-difference equivalent} to the other.~\footnote{Note that in any two ($w$-scaled) payoff-difference-equivalent games, every player achieves exactly the same ($w$-scaled) expected regrets with respect to any (fixed, possibly correlated) joint-mixed strategy.}
\end{definition}

\begin{definition}
Let us say that a graph has {\em totally disconnected or open neighborhoods\/} if for every node of the graph, every subgraph induced by the neighbors of the node is the empty graph; in other words, there is no edge connecting any pair of neighbors of any node in the graph; formally, if $E$ is the edge set, then for all $i$ and $j, k \in \mathcal{N}(i), j \neq k$, we have $(j,k) \notin E$.
\end{definition}
Some simple examples of graphs with totally disconnected neighborhoods are trees, cycles, and grids.
\begin{corollary}
\label{cor:symm}

Any {\em $w$-weighted Gibbs-potential\/} game with graph $G$ and
potential $\Psi$ is $w$-scaled payoff-difference equivalent to a {\em
  hyperedge-symmetric hypergraphical\/} game in which each hyperedge
is a (maximal) clique $C$ in $G$ and has as corresponding hypermatrix
the local potential associated to $C$ in $\Psi$. If, in addition, $G$
has totally disconnected neighborhoods then the equivalence is to {\em
  pairwise-symmetric graphical polymatrix\/} games.~\footnote{See
  Appendix (Missing Proofs) for all missing proofs.}

\end{corollary}
It is important to note the implication of the last corollary. In general, there is no reason to expect, {\em a priori}, just from the definition of a graphical potential game with, say for example, a tree graph that the differences in payoff matrix for a player would not be arbitrary functions of the action of the player and those of its neighbors. The corollary tells us that this is not possible in this case: the payoff hypermatrices difference must be \emph{sums} of simple \emph{pairwise} matrices, each being a simple 2-dimensional matrix depending on the actions of the player and one of its neighbors. The same holds for cycles, grids and similar structures, and their corresponding generalizations, including those to higher dimensions.

The following proposition completes the connection in the \emph{reverse} direction.
\begin{proposition}
\label{pro:gmg2gibbspg}
Any hyperedge-symmetric hypergraphical game is a graphical Gibbs-potential game.
\end{proposition}

\section{Smooth Best-Response Play and Graphical Gibbs-Potential Games}
\label{sec:sbr}

Let us consider a {\em sequential process of play} in which there is a pre-specified order by which the players play and at each time step $t$ exactly one player plays by choosing an action $x^t_i$.~\footnote{This process may be relaxed to allow certain dynamic variations in the order of play and some kinds of simultaneous moves. The process is also related to (smooth versions of) the {\em Cournot adjustment process with lock-in\/} and to {\em stochastic adjustment models\/} in the literature on learning in games~\citep{fudenberg99}.} In the sequel, let us take the sequence of play to be, without loss of generality, the sequence $1, 2, \ldots, n$. 
Let us say that a player has a  {\em (time-homogeneous, first-order) Markov playing scheme\/} if it has a (possibly randomized) policy, or plan, by which the agent selects an action based only on the {\em last actions\/} played by the other players; more formally, if the policy $p_i$ is a conditional probability distribution $p_i(x_i \mid x_{-i})$ such that if player $i$ is to play at time $t+1$ and $x^t_{-i}$ are the {\em last\/} joint-actions that player $i$ observed the others take, then player $i$ chooses to play action $x^{t+1}_i = x_i$ with probability $p_i(x_i \mid x^t_{-i})$ (i.e., $x^{t+1}_i \sim p_i(. \mid x^t_{-i})$).~\footnote{Although not pursued in this paper, more complex playing scheme could in principle be considered.} For a graphical game, let us further say that the playing scheme is {\em local\/} if it {\em only\/} depends on the last actions of its {\em neighbors\/} in the game graph. 
\begin{property}
\label{pr:gs}
In a graphical game with graph $G$, the sequential process of play generated by Markov playing schemes of the type described above, and that are local with respect to $G$, is equivalent to realizations generated by running the Gibbs sampler~\citep{geman84} with conditional distributions given by the individual player's playing scheme.
\end{property}
For every round $r = 1, 2, \ldots$ composed of consecutive time steps of length $n$, denote by $z^r = (z^r_1,\ldots,z^r_n)$ the {\em play (joint-action) outcome\/} at round $r$, so that for all $i$, $z^r_i = x^{(r-1)n+i}_i$ (i.e., the joint-action generated during round $r$ by the sequential process of play). Let us refer to the total sequence generated by $z^r$ for $r = 1, 2, \ldots, T$ as the {\em empirical play\/} up to round $T$ (i.e., after $n T$ rounds). Let us say that an empirical play {\em (conditionally) converges\/} starting from an initially assigned play $x^0$ if the empirical joint-probability distribution defined by the empirical play converges (almost surely) to some joint probability distribution as we let the sequential process of play run for an infinite number of rounds (i.e., for every joint-action $x$, if we denote the empirical distribution of play after $T$ rounds as $\widehat{P}^T(x) \equiv \frac{1}{T} \sum_{r=1}^T \indicator{z^r = x}$, then $\widehat{P}^\infty(x) \equiv \lim_{T \to \infty} \widehat{P}^T(x)$ exists, with probability one.) Let us refer to a set of playing schemes, one for each player, as a {\em playing procedure\/} for the game. Furthermore, let us say that a playing procedure is {\em consistent\/} (or {\em globally convergent\/}) if the empirical play generated is convergent to the {\em same\/} 
joint distribution
$\widehat{P}^\infty(x)$ from {\em any\/} initial joint play $x^0$.

Let us say that a player $i$ with payoff function $M_i$ uses a {\em smooth best-response (SBR) playing scheme\/} $p_i$ with respect to a (conditional) preference-order-preserving transform 
$f_i : \R \times A_{-i} \to \R$,
if for all $x_{-i}$, we have $p_i(x_i \mid x_{-i}) \propto \exp(f_i(M_i(x_i,x_{-i}), x_{-i})$.
If, instead, the $f_i$'s above are such that the ratio $p_i(x_i \mid x_{-i})/p_i(x'_i \mid x_{-i}) = \exp(f_i(M_i(x_i,x_{-i}) - M_i(x'_i,x_{-i}),x_{-i}))$, then let us call the scheme a {\em smooth best-response-difference (SBRD) playing scheme}.~\footnote{This condition implies that $f_i(v,x_{-i}) = - f_i(-v,x_{-i})$ and $p_i(x_i \mid x_{-i}) = 1/(\sum_{x'_i} \exp(f_i(M_i(x'_i,x_{-i}) - M_i(x_i,x_{-i}),x_{-i}))$.}
Recall from the previous discussion that, for a graphical game, a player's (playing) scheme is {\em local\/} with respect to the game graph, if it only depends on the actions of its neighbors in that graph. In the case of SBR and SBRD schemes, in general, this requires that the domain of the second argument to $f_i$ be $A_{\mathcal{N}(i)}$. Finally, let us say that the game has a {\em (local, Markov) playing procedure\/} if every player uses a (local, Markov) playing scheme.

The following theorems and corollaries provide another characterization of graphical potential games. 
\begin{theorem}
Any graphical game with graph $G$ that has a consistent (Markov) local SBR playing procedure is an ordinal Gibbs-potential game with graph $G$. If, in particular, the transform $f_i$ used by each player $i$ is (unconditionally) linear with weight $1/w_i$, then the game is a $w$-weighted Gibbs-potential game. If the playing procedure is SBRD, instead of SBR, then the game is a $f^{-1}$-transformed Gibbs-potential game, where $f^{-1}$ denote the inverses of the transforms $f$ used by the players for the SBRD scheme.
\end{theorem}
\begin{proof}
As noted in Property~\ref{pr:gs}, the playing procedure is equivalent to the Gibbs sampler. The Markov chain associated with the Gibbs sampler is regular because all the conditional probabilities associated with the players' schemes are positive. Because the procedure is convergent the limiting empirical distribution of play $\widehat{P}$ is the (unique) MRF consistent with the conditional probability distributions. Because $\widehat{P}$ is positive, by the Hammersley-Clifford Theorem (Theorem~\ref{thm:hc}), $\widehat{P}$ is also a Gibbs distribution with respect to the graph of the game $G$. Denote by $\widehat{\Psi}$ the Gibbs potential of $\widehat{P}$.
Also, denote by $\widehat{Z}_i(x_{-i}) \equiv \sum_{x'_i} \exp(\widehat{\Psi}(x'_i,x_{-i}))$ the normalizing constant for the marginal distribution over all the variables except $i$ with respect to $\widehat{P}$. Similarly, denote by $Z_i(x_{\mathcal{N}(i)}) \equiv \sum_{x'_i} \exp(f_i(M'_i(x'_i,x_{\mathcal{N}(i)}),x_{\mathcal{N}(i)}))$ the normalizing constant of the player's scheme. The conditionals of MRF $\widehat{P}$ must satisfy the following condition:
\begin{align*}
\frac{\exp(\widehat{\Psi}(x_i,x_{-i}))}{\widehat{Z}_i(x_{-i})} & =
\widehat{P}(X_i = x_i \mid X_{-i} = x_{-i}) \\
& = \widehat{P}(X_i = x_i \mid X_{\mathcal{N}(i)} =
x_{\mathcal{N}(i)})\\
& = p_i(x_i \mid x_{\mathcal{N}(i)}) \\
& = \frac{\exp(f_i(M'_i(x_i,x_{\mathcal{N}(i)}),x_{\mathcal{N}(i)}))}{Z_i(x_{\mathcal{N}(i)})} \; .
\end{align*}
From the last equality, we obtain that for all $x_i,x_{-i}$,
\[
f_i(M'_i(x_i,x_{\mathcal{N}(i)})) = \widehat{\Psi}(x_i,x_{-i}) + \ln(Z_i(x_{\mathcal{N}(i)})/\widehat{Z}_i(x_{-i})) \; .
\]
Because the second term on the right hand side of the last equation does not depend on $x_i$, we can obtain that for all $x_{-i}$ and every pair $x_i,x'_i$, 
\begin{align}
\nonumber f_i(M'_i(x_i,x_{\mathcal{N}(i)}),x_{\mathcal{N}(i)}) -
f_i(M'_i(x'_i,x_{\mathcal{N}(i)}),x_{\mathcal{N}(i)}) \\
\label{eqn:potdiff}
= \widehat{\Psi}(x_i,x_{-i}) - \widehat{\Psi}(x'_i,x_{-i}) \; .
\end{align}
Because the $f_i$'s are (strictly) monotonically increasing with their first argument, we have
\begin{align*}
& \widehat{\Psi}(x_i,x_{-i}) - \widehat{\Psi}(x'_i,x_{-i}) > 0 \\
& \iff  f_i(M'_i(x_i,x_{\mathcal{N}(i)}),x_{\mathcal{N}(i)}) - \\
& f_i(M'_i(x'_i,x_{\mathcal{N}(i)}),x_{\mathcal{N}(i)}) > 0\\
& \iff  f_i(M'_i(x_i,x_{\mathcal{N}(i)}),x_{\mathcal{N}(i)}) > 
f_i(M'_i(x'_i,x_{\mathcal{N}(i)}),x_{\mathcal{N}(i)})\\
& \iff  M'_i(x_i,x_{\mathcal{N}(i)}) > M'_i(x'_i,x_{\mathcal{N}(i)}) \\
& \iff  M'_i(x_i,x_{\mathcal{N}(i)}) - M'_i(x'_i,x_{\mathcal{N}(i)}) > 0 \; .
\end{align*}
This completes the proof of the first statement.

The second statement in the theorem follows by noting that in the case of (unconditionally) linear transforms with weight $w_i$, we have that the left-hand side of Equation~\ref{eqn:potdiff} above becomes
\(
f_i(M'_i(x_i,x_{\mathcal{N}(i)}),x_{\mathcal{N}(i)}) - f_i(M'_i(x'_i,x_{\mathcal{N}(i)}),x_{\mathcal{N}(i)}) = w_i M'_i(x_i,x_{\mathcal{N}(i)}) - w_i M'_i(x'_i,x_{\mathcal{N}(i)}) = w_i (M'_i(x_i,x_{\mathcal{N}(i)}) - M'_i(x'_i,x_{\mathcal{N}(i)})).
\)

The case of SBRD procedure is similar, but simpler:
\begin{align*}
& \exp(f_i(M_i(x_i,x_{\mathcal{N}(i)}) -
M_i(x'_i,x_{\mathcal{N}(i)}),x_{\mathcal{N}(i)}))\\
& = \frac{p_i(x_i \mid x_{\mathcal{N}(i)})}{p_i(x'_i \mid
  x_{\mathcal{N}(i)})}\\
& = \exp(\widehat{\Psi}(x_i,x_{-i}) - \widehat{\Psi}(x'_i,x_{-i})) \; .
\end{align*}
Hence, we have $M_i(x_i,x_{\mathcal{N}(i)}) - M_i(x'_i,x_{\mathcal{N}(i)}) = f^{-1}_i(\widehat{\Psi}(x_i,x_{-i}) - \widehat{\Psi}(x'_i,x_{-i}),x_{\mathcal{N}(i)})$. Using the definition of a graphical transformed potential game completes the proof.
\end{proof}

The next proposition completes the characterization of transformed Gibbs-potential games as \emph{exactly} those that have consistent local Markov SBRD playing procedures.
\begin{proposition}
\label{prop:sbr}
Any $f$-transformed Gibbs-potential game with graph $G$ has a consistent Markov SBRD playing procedure that is local with respect to $G$ and each player uses the inverse transforms of $f$ in their playing scheme. If, in particular, the game is a $w$-weighted Gibbs-potential game, then it also has an SBR procedure with the same properties.
\end{proposition}
Whether the last theorem extends to graphical {\em ordinal\/} potential games is left open.

\paragraph{Closing Remarks on Related Work.} A connection, obtained independently from the one in this paper, between potential games and MRFs appeared in previous
work in economics and game
theory~\citep{ellison93,Blume1993387,kandori93,young93,Blume1995111},
and more recently in the theoretical CS
community~\citep{Montanari09,Montanari12112010}. 

In all that previous
work, the interest and goal for making the connection is
different than the one here. Roughly speaking, their connection was via very specific ``game-playing
rules'' (e.g., the logit rule and its limiting version, the so called
myopic/deterministic/strict best-response rule) on very specific
classes of games (e.g.,
coordination games). Their interest was the nature of the steady state
induced by playing some
specific-type of rule in some specific kind of game, and sometimes
also establishing bounds on the convergence rate to steady state
behavior~\citep{Montanari09,Montanari12112010}. That previous work showed that the steady state of running such a rule in
a specific type of graphical game called a \emph{local interaction game} is an MRF with respect to the graph
of the specific type of game. 

The results presented here are more general in the type of game-playing
procedures considered for the characterization. This work does not study
limiting cases such as (myopic/deterministic/strict) best-response,
nor does it study convergence rates. 
\begin{center}
{\em No previous work provides a connection in
the other direction: if certain type of game-playing rule converges,
then the
players must be embedded in some
graphical \emph{potential} game.} 
\end{center}

From a technical perspective, by invoking known results from
probabilistic graphical models (e.g., Hammersley-Clifford Theorem~\citep{hammersley71,besag74} and
convergence of the Gibbs sampler~\citep{geman84}), the derivation of our results
simplified considerably. Interestingly, the proofs for the previous
results appear to be specific versions of those typically used to prove
the ``hard'' direction of
the Hammersley-Clifford Theorem.~\footnote{Also, in some cases, the proofs in previous work use more
sophisticated mathematical tools because they want to understand what
happens to the unique stable distribution as the best-response becomes
deterministic in the limit. They use a dynamical-system
  framework and apply tools such as the Strong Ergodic Theorem to
  prove the existence of the limit of the stable distributions, as the
  best-response rule becomes more myopic/deterministic/strict. The typical reference
  given is~\citet{karlin75}.}

\appendix

\section{Appendix: Missing Proofs}
\label{app:proofs}

This appendix contains all the proofs left out of the main body of the paper.

\subsection{Proof of Corollary~\ref{cor:symm}}

The proof follows from Theorem~\ref{the:gibbspot} by applying the respective definitions. For the last statement of the corollary, in particular, note that the maximal cliques of a graph with disconnected neighborhoods are exactly the edges of the graph. Hence, each local potential is pairwise in that case, and the payoff difference is the sum of differences of local pairwise potentials, as is the case in any graphical polymatrix games, by definition.

\subsection{Proof of Proposition~\ref{pro:gmg2gibbspg}}

For each hyperedge $C$ in the hypergraphical game, define a corresponding local potential $\phi_C(x_C) \equiv M'_C(x_C)$ as the local-clique payoff matrix for $C$. Defining the Gibbs potential using those local potentials and with the graph being the primal graph of the hypergraph, which is an undirected graph with the same vertex set as the hypergraph but where there is an edge between two nodes if there is a hyperedge containing both nodes.~\footnote{Note that in this case, because we are going in the ``reverse'' direction, local potential functions do not need to be over {\em maximal\/} cliques.}

\subsection{Proof of Proposition~\ref{prop:sbr}}

Let $\Psi$ be the Gibbs potential of the game. Define the SBRD playing scheme for each player using the inverse $f^{-1}_i$ of the corresponding game transform, such that
\begin{align*}
& \frac{p_i(x_i \mid x_{\mathcal{N}(i)})}{p_i(x'_i \mid x_{\mathcal{N}(i)})} \\
& \equiv \exp(f^{-1}_i(M_i(x_i,x_{\mathcal{N}(i)}) - M_i(x'_i,x_{\mathcal{N}(i)}),x_{\mathcal{N}(i)})) \\
& = \exp(\widehat{\Psi}(x_i,x_{-i}) - \widehat{\Psi}(x'_i,x_{-i})) \; 
\end{align*}
where the second equality follows from the definition of transformed potential games. The scheme is local by construction. Because $\Psi$ is a Gibbs potential, by the Hammersley-Clifford Theorem (Theorem~\ref{thm:hc}), there is a positive joint (global) MRF consistent with the conditional distributions induced by the scheme. Hence, running the Gibbs sampler will always converge, regardless of initial conditions. Because the playing procedure is equivalent to running the Gibbs sampler (Property~\ref{pr:gs}), the corresponding empirical distribution of play will always converge to the same Gibbs distribution with potential $\Psi$. Thus, the procedure will be consistent.

In the special case of a $w$-weighted Gibbs-potential game, then each player $i$'s scheme takes the form
\[
p_i(x_i \mid x_{\mathcal{N}(i)}) = \frac{\exp(M_i(x_i,x_{\mathcal{N}(i)})/w_i)}{\sum_{x'_i} \exp(M_i(x'_i,x_{\mathcal{N}(i)})/w_i)}
\]
which corresponds to an SBR with a (unconditionally) linear transform with weight $1/w_i$.


\bibliographystyle{plainnat}
\bibliography{games}

\end{document}